\documentclass[aps,prd,twocolumn,superscriptaddress,reprint]{revtex4-2}

\usepackage[utf8]{inputenc}
\usepackage[T1]{fontenc}
\usepackage{amsmath,amssymb,amsfonts,amsthm,bm}
\usepackage{physics}
\usepackage{mathtools}
\usepackage{xcolor}
\usepackage[colorlinks=true,linkcolor=blue,citecolor=blue,urlcolor=blue]{hyperref}

\newtheorem{theorem}{Theorem}
\newtheorem{proposition}[theorem]{Proposition}
\newtheorem{corollary}[theorem]{Corollary}

\theoremstyle{remark}

\begin{document}

\title{Modular Lower Bounds on Reeh--Schlieder State Preparation}

\author{Javier Blanco-Romero}
\email{frblanco@pa.uc3m.es}
\affiliation{Department of Telematic Engineering, Universidad Carlos III de Madrid, Legan\'es, Madrid, 28911, Spain}
\author{Florina Almenares Mendoza}
\affiliation{Department of Telematic Engineering, Universidad Carlos III de Madrid, Legan\'es, Madrid, 28911, Spain}

\date{\today}

\begin{abstract}
The Reeh--Schlieder theorem says that every target vector can be approximated from the vacuum by an operator localized in an arbitrarily small spacetime region, but it gives no quantitative cost for doing so. This note isolates a standard Tomita--Takesaki estimate as a model-independent preparation bound. Targets with deeply negative modular energy require large local operators.

After rescaling such an operator to a physical contraction, the same estimate becomes a lower bound on postselection overhead. In geometries where the modular Hamiltonian is known, the bound becomes explicit. Bisognano--Wichmann turns it into a boost energy statement for wedges, and the Casini--Huerta--Myers formula gives a stress-tensor version for bounded regions of conformal field theories. Local unitaries can only reach states of nonnegative modular energy. Negative modular sectors require nonunitary or postselected outcomes, giving a preparation cost bound that complements vacuum embezzlement in type III local algebras.
\end{abstract}

\maketitle

\section{Introduction}

Let $\mathcal{O}$ be a bounded open region of Minkowski space. In the Haag--Kastler framework a quantum field theory assigns to it a von Neumann algebra $\mathcal{A}(\mathcal{O})$ acting on the vacuum Hilbert space~\cite{haag1964algebraic,haag2012local}. The Reeh--Schlieder theorem states that~\cite{reeh1961bemerkungen}
\begin{equation}
\overline{\mathcal{A}(\mathcal{O})\Omega}=\mathcal{H}.
\label{eq:reeh_schlieder}
\end{equation}
Every vector in the Hilbert space can be approximated by acting on the vacuum inside $\mathcal{O}$ alone.

This is not a statement about efficient preparation. The theorem gives no bound on $\|A\|$, no bound on the energy of $A\Omega$, and no constructive way to find $A$. The gap matters because local algebras in QFT are type III factors, which support state-conversion phenomena absent from finite tensor-product systems. Relativistic quantum fields can act as universal entanglement embezzlers, drawing entanglement from the vacuum while disturbing a reference state arbitrarily little~\cite{van2024relativistic,van2024embezzlement}.

The natural next question is not whether local state preparation is possible. Reeh--Schlieder already answers that. The question is how expensive the local operator must be.

We ask one basic question. If $A\in\mathcal{A}(\mathcal{O})$ prepares $\xi=A\Omega$, how large must $\|A\|$ be? Modular theory supplies a canonical lower bound, controlled by the signed expectation of the vacuum modular Hamiltonian of $\mathcal{O}$ in the target state. In geometries where this Hamiltonian is explicit, the bound becomes a quantitative statement about boost energy in wedges and about a weighted stress-tensor integral in conformal causal diamonds.

\section{The modular bound}

Reeh--Schlieder establishes that $\Omega$ is cyclic and separating for $\mathcal{A}(\mathcal{O})$. Cyclicity is the statement that $\mathcal{A}(\mathcal{O})\Omega$ is dense in $\mathcal{H}$, as expressed by Eq.~\eqref{eq:reeh_schlieder}. The separating property says that $A\mapsto A\Omega$ is injective, and follows because $\Omega$ is also cyclic for the commutant $\mathcal{A}(\mathcal{O})'\supset\mathcal{A}(\mathcal{O}')$. The Tomita operator is the antilinear map
\begin{equation}
S_{\mathcal{O}}(A\Omega)=A^*\Omega,\qquad A\in\mathcal{A}(\mathcal{O}),
\end{equation}
which is well defined precisely because $\Omega$ is separating.
This map is in general unbounded but closable. Let $\overline S_{\mathcal{O}}$ denote its closure. Its polar decomposition is
\begin{equation}
\overline S_{\mathcal{O}}=J_{\mathcal{O}}\Delta_{\mathcal{O}}^{1/2},
\end{equation}
where $J_{\mathcal{O}}$ is the modular conjugation, an antiunitary involution, and $\Delta_{\mathcal{O}}$ is the positive modular operator~\cite{takesaki2003theory,witten2018aps}. We write
\begin{equation}
K_{\mathcal{O}}=-\log\Delta_{\mathcal{O}}.
\end{equation}

The Tomita operator $S_{\mathcal{O}}$ compares a local excitation of the vacuum with the excitation produced by the adjoint operator. The polar decomposition is more than a normal form. Positivity of $\Delta_{\mathcal{O}}$ makes its imaginary powers $\Delta_{\mathcal{O}}^{it}$ unitary, and the Tomita--Takesaki theorem shows that conjugation by them preserves the same local algebra. Thus the state-algebra pair $(\mathcal{A}(\mathcal{O}),\Omega)$ canonically determines a one-parameter automorphism group
\begin{equation}
\sigma_t(A)=\Delta_{\mathcal{O}}^{it}A\Delta_{\mathcal{O}}^{-it}.
\end{equation}
This is the modular flow, and $t$ is modular time. Writing $\Delta_{\mathcal{O}}=e^{-K_{\mathcal{O}}}$ gives $\Delta_{\mathcal{O}}^{it}=e^{-itK_{\mathcal{O}}}$, so $K_{\mathcal{O}}$ generates this intrinsic time evolution exactly as an ordinary Hamiltonian generates ordinary time evolution. What singles out this generator is that $\Omega$ is an equilibrium state for it. The correlator $\langle\Omega|A\,\sigma_t(B)|\Omega\rangle$ has the same analytic continuation and imaginary-time periodicity as a Gibbs state, which is the Kubo--Martin--Schwinger (KMS) condition~\cite{haag1967equilibrium}. In this precise sense the vacuum looks thermal when probed by the algebra $\mathcal{A}(\mathcal{O})$, with $K_{\mathcal{O}}$ playing the role of the Hamiltonian for the corresponding modular clock. In the geometric instances of Sec.~\ref{sec:geometric} this identification becomes literal, with $K_{\mathcal{O}}$ a boost generator for wedges and a weighted energy density for CFT balls. The antiunitary factor $J_{\mathcal{O}}$ in the polar decomposition maps $\mathcal{A}(\mathcal{O})$ onto its commutant, relating the region to its causal complement.

\begin{proposition}[Modular lower bound]
\label{prop:modular_bound}
Let $A\in\mathcal{A}(\mathcal{O})$ and set $\xi=A\Omega$. Then
\begin{equation}
\|A\|\ge \|\Delta_{\mathcal{O}}^{1/2}\xi\|.
\label{eq:modular_norm_bound}
\end{equation}
If $\xi\neq 0$, $\hat\xi=\xi/\|\xi\|$, and $\langle K_{\mathcal{O}}\rangle_{\hat\xi}$ is finite, then
\begin{equation}
\|A\|\ge \|\xi\|\exp\!\left[-\frac12\langle K_{\mathcal{O}}\rangle_{\hat\xi}\right].
\label{eq:jensen_bound}
\end{equation}
\end{proposition}

\begin{proof}
For vectors of the form $A\Omega$, the closed extension of the Tomita operator satisfies
\begin{equation}
\overline S_{\mathcal{O}}A\Omega=A^*\Omega.
\end{equation}
Using the polar decomposition of $\overline S_{\mathcal{O}}$, we get
\begin{equation}
\|A^*\Omega\|=\|J_{\mathcal{O}}\Delta_{\mathcal{O}}^{1/2}A\Omega\|=\|\Delta_{\mathcal{O}}^{1/2}\xi\|.
\end{equation}
Since $\|A^*\Omega\|\le \|A^*\|=\|A\|$, Eq.~\eqref{eq:modular_norm_bound} follows for every bounded local $A$.

For the second claim, assume the stated expectation of $K_{\mathcal{O}}$ is finite. The functional calculus for the self-adjoint operator $K_{\mathcal{O}}$ associates to the state $\hat\xi$ a probability measure $\mu_{\hat\xi}$ on its spectrum such that $\langle f(K_{\mathcal{O}})\rangle_{\hat\xi}=\int f(x)\,d\mu_{\hat\xi}(x)$ for any bounded Borel $f$. Applying this with $f(x)=e^{-x}$ and using $\Delta_{\mathcal{O}}=e^{-K_{\mathcal{O}}}$,
\begin{equation}
\|\Delta_{\mathcal{O}}^{1/2}\hat\xi\|^2
=\langle e^{-K_{\mathcal{O}}}\rangle_{\hat\xi}
=\int e^{-x}\,d\mu_{\hat\xi}(x).
\end{equation}
Since $x\mapsto e^{-x}$ is convex, Jensen's inequality gives
\begin{equation}
\int e^{-x}\,d\mu_{\hat\xi}(x)
\ge \exp\!\left[-\int x\,d\mu_{\hat\xi}(x)\right]
=\exp[-\langle K_{\mathcal{O}}\rangle_{\hat\xi}],
\end{equation}
which proves Eq.~\eqref{eq:jensen_bound}.
\end{proof}

The proposition says that the Hilbert-space norm of the target is not the only scale relevant to local preparation. The same operator that creates $A\Omega$ also has an adjoint action $A^*\Omega$, and Tomita theory identifies its size with the modular norm $\|\Delta_{\mathcal{O}}^{1/2}\xi\|$. Reeh--Schlieder allows approximation of any target vector, but vectors with large negative modular energy require large local operators.

For a normalized target, the elementary bound $\|A\|\ge \|A\Omega\|$ already gives $\|A\|\ge 1$. If $\langle K_{\mathcal{O}}\rangle_{\hat\xi}$ is positive, Eq.~\eqref{eq:jensen_bound} is weaker than this trivial estimate. If $\langle K_{\mathcal{O}}\rangle_{\hat\xi}\le -M$, it gives
\begin{equation}
\|A\|\ge \|\xi\|e^{M/2}.
\end{equation}
Thus the exponential cost comes from the negative modular sector. The sharper norm bound \eqref{eq:modular_norm_bound} may see more of the spectrum, but the expectation value bound has its physical bite when the signed modular energy is negative.

The estimate is an elementary consequence of $S(A\Omega)=A^*\Omega$ and Jensen's inequality. It turns standard modular structure into a direct lower bound for local preparation and postselection in the type III setting.

The quantity $\langle K_{\mathcal{O}}\rangle_{\hat\xi}$ is a signed expectation value of the vacuum modular Hamiltonian. It is not the Araki relative entropy $S(\omega_{\hat\xi}\,\|\,\omega_\Omega)$~\cite{araki1977relative}, which is a different, nonnegative object built from the relative modular operator. Negative modular energy therefore does not mean negative total Minkowski energy. It means negative energy with respect to the modular clock attached to the chosen algebra and the vacuum.

\section{Postselection overhead}

The operator norm has a direct operational meaning. Recall that a contraction is a bounded operator $B$ with $\|B\|\le 1$, equivalently $\|B\psi\|\le\|\psi\|$ for every $\psi$; it need not be unitary or even isometric. Any nonzero bounded local operator $A$ can be rescaled to the norm-one contraction
\begin{equation}
B=\frac{A}{\|A\|}.
\end{equation}
More generally, any local contraction $B$ satisfies $B^*B\le I$. Hence it can serve as the Kraus operator for the success outcome of a local two-outcome measurement. The other outcome has Kraus operator $(I-B^*B)^{1/2}$, which also lies in $\mathcal{A}(\mathcal{O})$. Acting on the vacuum, the probability of the success outcome is
\begin{equation}
p_{\rm succ}(B)=\|B\Omega\|^2.
\end{equation}

\begin{corollary}[Modular overhead]
\label{cor:overhead}
Let $B\in\mathcal{A}(\mathcal{O})$ be a contraction with $B\Omega\ne 0$, and let $\hat\xi=B\Omega/\|B\Omega\|$. Then
\begin{equation}
p_{\rm succ}(B)
\le \|\Delta_{\mathcal{O}}^{1/2}\hat\xi\|^{-2}.
\label{eq:success_bound}
\end{equation}
If $\langle K_{\mathcal{O}}\rangle_{\hat\xi}$ is finite, the expected number of trials obeys
\begin{equation}
p_{\rm succ}(B)^{-1}
\ge \exp[-\langle K_{\mathcal{O}}\rangle_{\hat\xi}].
\label{eq:overhead_bound}
\end{equation}
\end{corollary}

\begin{proof}
Apply Eq.~\eqref{eq:modular_norm_bound} to $B$:
\begin{equation}
1\ge \|B\|\ge \|B\Omega\|\,\|\Delta_{\mathcal{O}}^{1/2}\hat\xi\|.
\end{equation}
Squaring gives the success bound. Jensen's inequality gives the overhead bound.
\end{proof}

Thus a modular lower bound on $\|A\|$ is also a lower bound on postselection cost. This connects the AQFT statement to the usual language of quantum information, where postselection can make transformations possible at exponential cost~\cite{aaronson2005quantum}. Here we use that connection only as a resource statement. If a computational encoding forces the successful outcome into modular energies of order $-M(n)$, then any local postselected preparation in this model needs at least $e^{M(n)}$ trials. Turning this into a complexity theorem requires a specified encoding and readout model, but the physical obstruction is already visible at the level of local state preparation.

\section{Geometric realizations}\label{sec:geometric}

The bound becomes useful when $K_{\mathcal{O}}$ is explicit. The modular Hamiltonian is the clock naturally attached to the chosen algebra. It need not be the ordinary Hamiltonian that translates all of Minkowski time. For special regions, however, the clock has a geometric meaning.

\subsection{Wedges}

Let
\begin{equation}
W_R=\{x\in\mathbb{R}^{d+1}:x^1>|x^0|\}
\end{equation}
be the right Rindler wedge. The Bisognano--Wichmann theorem~\cite{bisognano1976duality} identifies the modular operator of $(\mathcal{A}(W_R),\Omega)$ with the Poincar\'e boost generator that preserves the wedge. Concretely, $K_{\rm boost}$ is the conserved charge associated with $x^1$-boosts,
\begin{equation}
K_{\rm boost}=\int_{x^0=0}d^d x\;x^1\,T_{00}(\bm{x}),
\end{equation}
and the convention
\begin{equation}
\Delta_{W_R}=e^{-2\pi K_{\rm boost}}
\end{equation}
makes modular flow the one-parameter subgroup of Lorentz boosts in the $x^0$-$x^1$ plane. Proposition~\ref{prop:modular_bound} gives
\begin{equation}
\|A\|\ge \|\xi\|\exp[-\pi\langle K_{\rm boost}\rangle_{\hat\xi}].
\label{eq:wedge_bound}
\end{equation}

Here modular time is Rindler time. It is the time measured by uniformly accelerated observers who remain inside the wedge. The corresponding generator is a boost, not the global Hamiltonian. It is nevertheless an energy in the operational sense relevant to those observers, since it is the conserved quantity that translates their time.

The sign matters. States with positive boost energy produce no large lower bound from Eq.~\eqref{eq:wedge_bound}. States with large negative boost energy do. A right wedge operation cannot signal across the horizon, but Reeh--Schlieder says that its action on the entangled vacuum is dense in the Hilbert space. Outcomes that run strongly against Rindler time require exponentially small postselection probability or exponentially large operator norm. The obstruction is therefore not ordinary energy by itself. It is energy measured by the modular clock of the region.

\subsection{CFT balls}

Wedges are unbounded. A bounded region is more relevant for local preparation. In a $d+1$-dimensional conformal field theory, the causal diamond $\mathcal{D}_R$ of a ball of radius $R$ has a local modular Hamiltonian. Casini, Huerta, and Myers showed that $K_{\mathcal{D}_R}=-\log\Delta_{\mathcal{D}_R}$ equals the stress-tensor integral on the full Hilbert space~\cite{casini2011towards}:
\begin{equation}
K_{\mathcal{D}_R}
=2\pi\int_{|\bm{x}|<R} d^{d}x\,
 w_R(\bm{x})\,T_{00}(\bm{x}),
\label{eq:chm}
\end{equation}
where
\begin{equation}
w_R(\bm{x})=\frac{R^2-|\bm{x}|^2}{2R}
\end{equation}
and Proposition~\ref{prop:modular_bound} becomes the explicit preparation bound
\begin{align}
\|A\|&\ge \|\xi\|\exp\!\left[-\frac12\langle K_{\mathcal{D}_R}\rangle_{\hat\xi}\right] \\
&= \|\xi\|\exp\!\left[-\pi\int_{|\bm{x}|<R} d^{d}x\,
w_R(\bm{x})\,\langle T_{00}(\bm{x})\rangle_{\hat\xi}\right].
\label{eq:cft_general}
\end{align}
In particular, any family of target states with $\langle K_{\mathcal{D}_R}\rangle_{\hat\xi}\le -M$ requires
\begin{equation}
\|A\|\ge \|\xi\|e^{M/2},
\qquad
p_{\rm succ}\le e^{-M}.
\label{eq:cft_exponential}
\end{equation}

This is the ball version of the wedge obstruction. The modular clock is now the one-parameter conformal isometry of Minkowski space that maps $\mathcal{D}_R$ to itself, a conformal Killing flow whose orbits foliate the diamond and reduce to the time translations of an inertial observer at its center. On the central time slice its generator is the spatially weighted energy in Eq.~\eqref{eq:chm}, with lapse $w_R(\bm{x})$ that vanishes at the boundary and peaks at the center. Negative energy near the center therefore contributes strongly to negative modular energy and drives a large exponential cost; negative energy very close to the boundary is suppressed in this bound, not because it is physically free, but because the conformal Killing flow becomes null at the entangling surface and the modular Hamiltonian has little leverage there.

This negative modular sector is not merely formal. Local energy density in QFT is not a positive operator. Even after smearing with a smooth test function, there are physical states whose renormalized energy density is negative in the sampled region~\cite{epstein1965nonpositivity,witten2018aps}. In free-field language these are usually discussed as squeezed-state negative-energy pulses. The pulse is not free energy from nowhere. It is accompanied by compensating positive energy elsewhere, so the global Hamiltonian can remain positive.

The ball modular Hamiltonian can nevertheless be negative because it samples energy with a position-dependent lapse. Suppose, schematically, that a target state has
\begin{equation}
\langle T_{00}(\bm{x})\rangle_{\hat\xi}<0
\quad\hbox{in a small region near }\bm{x}=0,
\end{equation}
and that the compensating positive energy lies closer to $|\bm{x}|=R$. Since $w_R(0)=R/2$ while $w_R(\bm{x})\to0$ at the boundary, the negative contribution can dominate the weighted integral even when the ordinary total energy is nonnegative. In a simple two-pulse model, with negative energy $-E_-$ localized near radius $r_-$ and positive energy $E_+$ localized near radius $r_+$, the sign of the modular energy is controlled by
\begin{equation}
\langle K_{\mathcal{D}_R}\rangle_{\hat\xi}
\simeq 2\pi\left[-w_R(r_-)E_-+w_R(r_+)E_+\right].
\end{equation}
Thus $r_-\ll R$ and $r_+\lesssim R$ can give $\langle K_{\mathcal{D}_R}\rangle_{\hat\xi}\le -M$ without requiring negative total energy. Eq.~\eqref{eq:cft_exponential} then says that any local Kraus operator preparing this state from the vacuum succeeds with probability at most $e^{-M}$. The exponential barrier is therefore attached to a known family of QFT states. The relevant energy is not the total Hamiltonian energy, but the energy distribution as seen by the diamond's modular clock.

Both cases share the same anatomy. The obstruction is signed energy seen through a local modular clock, namely Rindler boost time for the wedge and conformal time of the diamond (with lapse $w_R(\bm{x})$) for the ball. What singles out these two regions is that their modular Hamiltonians are local and explicit, so modular flow acts as an honest Hamiltonian evolution. For a generic region, the modular Hamiltonian is nonlocal and not known in closed form. In a type~I approximation one would write it in terms of reduced density matrices of the vacuum on the region and its complement, but this formal expression does not by itself provide a usable local generator; it still requires access to the spectral decomposition of the reduced state.

\section{Comparison with vacuum embezzlement}

Universal embezzlement~\cite{van2024relativistic,van2024embezzlement} shows that relativistic local algebras allow powerful state conversion in principle. Proposition~\ref{prop:modular_bound} addresses a complementary operational question. If the conversion is represented, on the original vacuum Hilbert space, by a single bounded local Kraus operator, how large must that operator be, or how rare is the corresponding measurement outcome?

The unitary case marks the deterministic boundary. If $U\in\mathcal{A}(\mathcal{O})$ is unitary, then $\|U\|=1$ and $\xi=U\Omega$ has unit norm, so Eq.~\eqref{eq:jensen_bound} gives
\begin{equation}
\langle K_{\mathcal{O}}\rangle_{U\Omega}\;\ge\;0.
\label{eq:unitary_floor}
\end{equation}
A local unitary may rearrange the vacuum, but it cannot produce a state of negative modular energy for the same local algebra.

This is the operational role of the bound. It locates where deterministic local dynamics stops. State preparation may still use measurement and conditioning, and a conditioned outcome is represented by a contraction rather than a unitary. If such an outcome targets a state with $\langle K_{\mathcal{O}}\rangle\le -M$, Corollary~\ref{cor:overhead} gives
\begin{equation}
p_{\rm succ}\le e^{-M}.
\end{equation}
The negative modular sector is therefore reachable only as a rare local measurement outcome, with the cost paid in postselection overhead.

\section{Conclusion}

Reeh--Schlieder reachability is exact in spirit and can be extravagant in cost. Modular theory gives a precise obstruction. If a local operator $A$ prepares $\xi=A\Omega$, then its norm is bounded below by the modular size of $\xi$. For negative modular energy the expectation-value bound becomes exponential, and after rescaling to a contraction it becomes a postselection overhead bound.

The estimate does not construct efficient Reeh--Schlieder approximants, make ordinary energy a universal complexity measure, or prove a complexity class separation. It identifies signed modular energy, relative to the chosen local algebra, as an obstruction to cheap local preparation. If future encodings of computational tasks force target states deep into negative modular sectors, the bound becomes an exponential resource barrier for those encodings. Universal embezzlement shows that relativistic local algebras allow strong state conversion in principle~\cite{van2024relativistic,van2024embezzlement}; once the conversion is represented by a bounded Kraus operator, modular geometry controls the success probability, and for local unitaries it forces the vacuum to nonnegative modular energy.

Type III vacuum entanglement makes strong state conversion possible. Modular geometry sets a lower bound on the local operator that performs it.

\begin{acknowledgments}
J.B.-R. is grateful to Juan Le\'on for illuminating conversations on the foundations of quantum mechanics and quantum field theory during his bachelor thesis, which shaped his interest in the subject. This work was supported by the Spanish Government under grant PID2023-148716OB-C33, funded by MICIU/AEI/10.13039/501100011033, as part of the DISCOVERY project (``DIstributed Smart Communications with Verifiable EneRgy-optimal Yields''). The Community of Madrid provided further support through the RAMONES-CM project (TEC-2024/COM-504), awarded by Order~5696/2024 of the General Directorate of Research and Technological Innovation.
\end{acknowledgments}

\bibliographystyle{apsrev4-2}
\bibliography{main}

\end{document}